\documentclass[11 pt]{article}

\usepackage{amsmath, amssymb}
\usepackage{graphicx}
\usepackage{xcolor}
\usepackage{ulem}
\usepackage{subcaption}
\usepackage[margin = 1in]{geometry}

\newcommand{\x}{\textbf{x}}
\renewcommand{\v}{\textbf{v}}
\renewcommand{\r}{\textbf{r}}

\newcommand{\Tp}{\mathcal{T}}
\newcommand{\I}{\mathcal{I}}

\newcommand{\rcap}{r_{\rm cap}}

\newcommand{\R}{\mathbb{R}}

\DeclareMathOperator{\E}{\mathbb{E}}

\newtheorem{theorem}{Theorem}
\newtheorem{lemma}{Lemma}
\newtheorem{corollary}{Corollary}
\newtheorem{remark}{Remark}

\newtheorem{proposition}{Proposition}

\newtheorem{assumption}{Assumption}

\newenvironment{proof}{\textit{Proof:}}{\hfill $\blacksquare$}

\title{Optimal Intermittent Sensing for Pursuit-Evasion Games}

 \author{Dipankar Maity \footnote{D. Maity is with the Department of Electrical and Computer Engineering, University of North Carolina at Charlotte,  NC, 28223, USA. Email:
 		{\small \texttt{dmaity@uncc.edu}}.\\
   The work of D. Maity was supported, in parts, by by  ARL grant ARL DCIST CRA W911NF-17-2-0181.} 
}
\date{}

\begin{document}

\maketitle

\begin{abstract}
     \textit{We consider a class of pursuit-evasion differential games in which the evader has continuous access to the pursuer's location, but not vice-versa. 
     There is a remote sensor (e.g., a radar station) that can sense the evader's location upon a request from the pursuer and communicate that sensed location to the pursuer. The pursuer has a budget on the total number of sensing requests. The outcome of the game is determined by the sensing and motion strategies of the players. We obtain an equilibrium sensing strategy for the pursuer and an equilibrium motion strategy for the evader.
     We quantify the degradation in the pursuer's pay-off due to its sensing limitations. }
\end{abstract}

\textit{\textbf{keywords -}}
Pursuit-evasion, intermittent sensing, self-triggered control.

\section{Introduction}

Pursuit-Evasion games \cite{isaacs1999differential}  have been applied to investigate a wide class of civilian and military applications involving multi-agent interactions in adversarial scenarios \cite{yan2016multi,  robotics9020047, inproceedings}. 
In its simplest form the game involves a pursuing player that is tasked to capture an evading player before either the evader reaches its destination or the pursuer runs out of fuel. 

While several variations ranging from complex dynamic models for the players (e.g., \cite{sun2017multiple}) to complex geometry of the environment (e.g., \cite{oyler2016pursuit}) to limited visibility of the players (e.g., \cite{bhattacharya2010existence}) have been considered, one of the prevailing assumptions have been the continuous sensing capability for the players, with the exception of \cite{aleem2015self, aleem2015bself,maity2016strategies, maity2016optimal} among few others. 
By `continuous sensing' we refer to the capability that enables the players to keep their sensors turned on continuously for the entire duration of the game.
Extensions of pursuit-evasion games in the context of sensing limitations have mainly considered limited sensing range \cite{bopardikar2007cooperative, durham2010distributed, shishika2021partial} and limited field of view \cite{gerkey2006visibility}, but the challenges associated with the lack of continuous sensing remain unsolved.  

In this work, we  revisit the classical pursuit evasion game in an obstacle-free environment where the pursuer does not have a continuous sensing capability. 
In particular, the pursuer relies on a remotely located sensor (e.g., a radar station) to sense the evader's position. 
Upon request, the remote sensor can perfectly sense the location of the evader and share it with the pursuer.\footnote{
One may alternatively also consider a scenario where the pursuer has an onboard sensor to measure the evader's location, however, it can only use the sensor intermittently due to, for example, energy and computational constraints.
} 
The communication channel between the pursuer and the remote sensor is assumed to be noiseless, instantaneous (i.e., no delay), and perfectly reliable (i.e., no packet losses). 
The pursuer intermittently requests the evader's location to update its pursuit strategy. 
Due to the resource (e.g., energy) constraints, the pursuer can only make a maximum of $n$ requests.
On the other hand, the evader is able to sense the pursuer continuously and is aware of the sensing limitation of the pursuer. 
The objective of this work is to analyze the game under this asymmetric sensing limitation and obtain the optimal instances for sensing.

Our prior works \cite{maity2016optimal, maity2016strategies} considered a linear quadratic differential game formulation where both the agents had sensing limitations. 
These works were later extended to discrete time \cite{maity2017linear}, infinite horizon \cite{maity2017asymptotic}, and recently to asset defense scenarios \cite{huang2021defending}. All these extensions rely on the linear-quadratic structure of the problem.
Closely related to the problem considered in this work are the works in \cite{aleem2015self} and \cite{aleem2015bself} where the authors considered a self triggered strategy to decide the sensing instances. 
Self-triggered strategies are conservative and often highly suboptimal. 
\cite{aleem2015bself} extends the work of \cite{aleem2015self} by incorporating noisy sensor measurements.
These works did not consider any budget on the number of sensing and therefore, the proposed sensing strategy is agnostic to the sensing budget $n$. 
With respect to these existing works, the contributions of this work are: (i) Formulating a sensing (and fuel) limited pursuit-evasion game and analyzing the equilibrium sensing and motion strategies for the players, (ii) Demonstrating the optimality of a `waiting strategy' for a pursuer equipped with only intermittent sensing, (iii)  Finding the required number of sensing $n_{\max}$ to ensure that the pursuer can perform the same as it does under continuous sensing, and (iv) Comparing the performance with existing work \cite{aleem2015self} to illustrate the sensing improvement achieved in our work.

The rest of the manuscript is organized as follows: In Sec.~\ref{sec:prelim} we describe our preliminary results and derive the required number of sensing to ensure that the pursuer performs the same as it does in the continuous sensing case. 
Building upon this result, in Sec.~\ref{sec:formulation} we present the sensing and fuel constrained game. 
In Sec.~\ref{sec:analysis}, we analyze the game and derive the associated value function. 
We further investigate how the value function behaves with respect to the sensing budget $n$ in Sec.~\ref{sec:senseN}, and finally, we conclude the manuscript in Sec.~\ref{sec:conslusions}.

\section{Existing and  Preliminary Results} \label{sec:prelim}
We consider a pursuit evasion game with simple motion:
    \begin{align} \label{eq:dynamics}
    \begin{split}
    \dot{\x}_p =\v_p,\\
    \dot{\x}_e =\v_e,
    \end{split}
\end{align}
where $\x_p(t) \in \R^2$, $\x_e(t) \in \R^2$ denote the locations and $\v_p(t) \in \R^2$, $\v_e(t) \in \R^2$ denote the velocities of the pursuer and the evader, respectively,  at time $t$. 
We assume that the game starts at time $t=0$ and both the players know the initial locations $\x_p(0)$ and $\x_e(0)$.
The maximum speeds for the pursuer and the evader are $1$ and $\nu$, respectively.\footnote{
In case the maximum speeds are $\v_{p,\max}$ and $\v_{e,\max}$, we shall use scaling and time dilation to transform the system to obtain maximum speeds $1$ and $\nu$, respectively.  
To that end, we define $\nu =\v_{e,\max}/\v_{p,\max}$, $c =\v_{p,\max}$, $\bar \x_i(t) = \x_i( t/c )$  and $\bar\v_i(t) =\v_i(t/c)/c$ for $i=p, e$. 
With this new scaling, we obtain $\dot{\bar{\x}}_i = \bar{\v}_i$ with $\|\bar\v_p(t)\| \le 1$ and $\|\bar\v_e(t)\|\le \nu$.} 
We consider a faster pursuer, i.e., $\nu <1$, to avoid a trivial game.

For a pursuer with capture radius $\rcap$, capture happens as soon as the distance between the pursuer and the evader becomes $\rcap$.
When both players can sense each other continuously, capture happens within a duration of $\frac{\rho_0 - \rcap}{1-\nu}$, where $\rho_0 \triangleq \|\x_p(0) - \x_e(0)\|$ is the initial distance between the players.
In this case, the equilibrium strategy for both the players is to move  along the pursuer's line of sight unit vector $\r(t) \triangleq \frac{\x_e(t)- \x_p(t)}{\|\x_e(t)- \x_p(t)\|}$ with maximum speeds for all $t$. 
One may verify that $\dot{\r}(t) = 0$ under this equilibrium. 
Consequently, the instantaneous heading angles for both the players do not change with time.  
It is noteworthy that \textit{the players do not need their opponent's location continuously to implement their equilibrium strategies}. 
Although it appears that the sensing capability is redundant in implementing the equilibrium strategies, however, the lack thereof is detrimental for the pursuer since the evader has an incentive to deviate from the above mentioned equilibrium strategy.
On the other hand, when the evader is incapable of continuous sensing, the pursuer still does not have any incentive to deviate from the above mentioned equilibrium strategy.
Therefore, a continuous sensing capability is crucial for the pursuer.


In this work we consider only an intermittent sensing capability for the pursuer. 
Similar sensing constrained pursuit-evasion games have been previously studied in \cite{aleem2015self} where the pursuer follows a self-triggered strategy to sense the evader intermittently and correct its heading angle. 
In \cite{aleem2015self}, the pursuer follows the pursuit strategy 
\begin{align} \label{eq:purserV}
   \v_p(t) = \r(t_k) \qquad \text{for all } t \in (t_k, t_{k+1}],
\end{align}
where $t_k$ denotes the $k$-th sensing instance. 
We denote the $0\text{-th}$ sensing instance to be the initial time, i.e., $t_0 = 0$.
A self-trigger function was designed in \cite{aleem2015self} to determine the sensing instances and a tight upper bound on the number of required sensing $n_{\max}$ was derived:
\begin{align}
    &n_{\max} = \bigg\lceil \frac{\log\rcap - \log\rho_0}{\log(h(\nu))}\bigg\rceil, \label{eq:nmax} \\
    &h(\nu) = 1 - \frac{\nu(1 - \nu)\sqrt{1 -\nu^2} - (1 - \nu)(1 - \nu^2)}{2\nu^2 -1}, \nonumber
\end{align}
where $\rho_0$ is the initial distance between the players.
One may verify that $1 > h(\nu) > \nu$ for all $\nu\in (0,1)$ and consequently, we get $n_{\max} \ge \big\lceil \frac{\log\rcap - \log\rho_0}{\log \nu}\big\rceil$.
Their proposed triggering strategy \cite[Eq. (6)]{aleem2015self} yields 
\[
t_{k+1} - t_k =  f(\nu) \|\x_e(t_k) - \x_p(t_k)\|,
\]
where $f(\nu) = \frac{\sqrt{1-\nu^2}}{\nu + \sqrt{1-\nu^2}}$. 
Notice that $f(\nu)$ decreases with $\nu$ and converges to $0$ as $\nu$ converges to $1$. 
In other words, the inter sensing time $(t_{k+1}- t_k)$ decreases with $\nu$. 
In the following we state a more efficient sensing strategy with two major properties: (i) The total required number of sensing is strictly less than the $n_{\max}$ given in \eqref{eq:nmax}, and  (ii) There is no degradation in pursuer's performance in comparison to the continuous sensing case, i.e., the capture time does not increase due to sensing limitation.
Before stating this main result for this section, we provide the following lemma that will be used for proving our result.

\begin{lemma} \label{lem:no-trigger}
    If the initial distance between the players is less than or equal to $ \frac{\rcap}{\nu}$, then there is no need for sensing and the equilibrium strategy for both the players is to move along the direction of $\r(0)$. \hfill $\triangle$
\end{lemma}

\begin{proof}
    Under the proposed strategies $\v_p(t) = \r(0)$ and $\v_e(t) = \nu \r(0)$, the capture time is $\frac{\rho_0 - \rcap}{1-\nu}$, where $\rho_0$ is the initial distance between the players.
    To proof that the proposed strategies form an equilibrium pair, we let the evader follow an arbitrary strategy and show that the capture time is strictly reduced. 
    To that end, note that $\x_p(t) = \x_p(0) + t\r(0)$ and $\x_e(t) = \x_e(0) + \int_0^t\v_e(s) ds$, and therefore, for all $t$,
    \begin{align*}
        \|\x_e(t) - \x_p(t)\| &= \|(\rho_0-t)\r(0) + \int_0^t\v_e(s) ds\| \\
        &\le |\rho_0 - t| + \nu t,
    \end{align*}
    where the equality holds if and only if $\v_e(t) = \nu\r(0) $ for all $t$.
    Consequently, at $t = \frac{\rho_0-\rcap}{1-\nu}$, we have $\|\x_e(t) - \x_p(t)\|\le \rcap $, which implies that capture must have happened before $t = \frac{\rho_0-\rcap}{1-\nu}$ unless the evader followed $\v_e(t) = \nu\r(0) $.
    In a similar fashion, one may also verify that there is no incentive for the pursuer to deviate from the strategy $\v_p(t) =\r(0)$.
\end{proof}

From Lemma~\ref{lem:no-trigger}, we notice that for any initial separation of $\rho_0 \le \frac{\rcap}{\nu}$,  the capture time is $\frac{\rho_0 - \rcap}{1-\nu}$.
This is the same capture time as what we also obtain from continuous sensing.
In other words, the pursuer's performance (measured by its ability to capture and the associated capture time) with or without the sensing capabilities are the same when $\rho_0\le \frac{\rcap}{\nu}$.
We now state the following Proposition on the number of required sensing for a given arbitrary initial distance $\rho_0$.

\begin{proposition}\label{prop:preResult}
Let the pursuer simply move to the last sensed location of the evader and request for a sensing as soon as it gets there. 
Then, the capture time is upper bounded by $\frac{\rho_0 -\rcap}{1-\nu}$ and the total number of required sensing is $\big\lfloor \frac{\log \rcap - \log \rho_0}{\log \nu} \big\rfloor$. \hfill $\triangle$
\end{proposition}

\begin{proof}
    Let the distance between the players at the $k$-th sensing moment be denoted by $\rho_k \triangleq \|\x_e(t_k) - \x_p(t_k)\|$. 
    The next sensing request happens when the pursuer reaches $\x_e(t_k)$.
    The time taken by the pursuer to reach the point $\x_e(t_k)$ is exactly $\rho_k$, i.e., $t_{k+1} - t_k = \rho_k$. 
    Therefore, the distance between the players at the $(k+1)$-th sensing instance is
    \begin{align*}
        \rho_{k+1} &= \|\x_e(t_{k+1}) - \x_p({t_{k+1}})\| = \|\x_e(t_{k+1}) - \x_e({t_{k}})\| \\
        &\le \nu (t_{k+1} - t_k) = \nu \rho_k \le \nu^{k+1} \rho_0.
    \end{align*}
    Therefore, with $n_{\max} = \big\lfloor \frac{\log \rcap - \log \rho_0}{\log \nu} \big\rfloor$, we obtain $\rho_{n_{\max}} \le \frac{\rcap}{\nu}$.
    At this point, we invoke Lemma~\ref{lem:no-trigger} to conclude that capture is inevitable without any further sensing.
    Furthermore, starting with an initial distance of  $\rho_{n_{\max}}$, the capture time is $\frac{\rho_{n_{\max}}-\rcap}{1-\nu}$. 
    Therefore, the total capture time is
    \begin{align*}
        T_{\rm capture} &= \frac{\rho_{n_{\max}}-\rcap}{1-\nu} + \sum\nolimits_{k=1}^{n_{\max}}(t_k - t_{k-1})\\
        & = \frac{\rho_{n_{\max}}-\rcap}{1-\nu} +\sum\nolimits_{k=1}^{n_{\max}}\rho_{k-1} \le \frac{\rho_0 - \rcap}{1-\nu},
    \end{align*}
    where we have used $t_0 = 0$ and $\rho_k \le \nu^k \rho_0$.
    This completes the proof.
\end{proof}

In Fig.~\ref{fig:compare}, we compare the number of required sensing proposed by \cite{aleem2015self} and by Proposition~\ref{prop:preResult} to demonstrate the sensing efficiency of our proposed strategy.
We notice that the difference between the $n_{\max}$ proposed in \cite{aleem2015self} and that in Proposition~\ref{prop:preResult} increases with $\nu$ and there is an order of magnitude difference (note the log scale on the $y$-axis) when $\nu$ is high.
\begin{figure}
    \centering
    \includegraphics[width = 0.6 \linewidth]{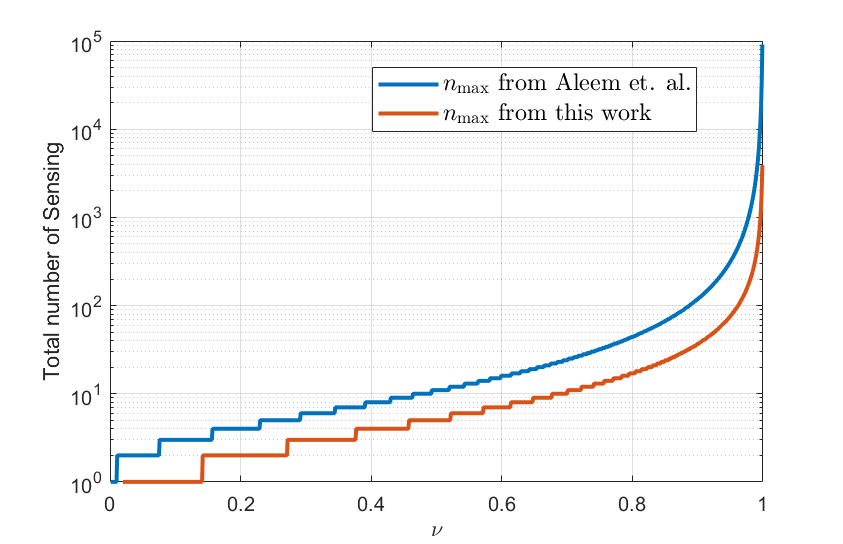}
    \caption{\small The required number of sensing vs. $\nu$. The $y$-axis being in $\log$ scale shows that the required number of sensing according to \cite{aleem2015self} is several times higher than what is found in Proposition~\ref{prop:preResult}.}
    \label{fig:compare}
    \vspace{-12 pt}
\end{figure}
Next, we comment on the time consumed (or distance traveled) by the pursuer following the strategy in Proposition~\ref{prop:preResult}. 
This will be helpful in our subsequent analysis when we consider constraints on the fuel consumed by the pursuer (or equivalently, constraints on the game duration). 
\begin{corollary}
Given $\rho_0$ and $\rcap$, let $n$ be the smallest integer such that $\rcap > \nu^n \rho_0$.
Then, $\max\{n\!-\!1,0\}$ number of sensings are required. 
Furthermore, the pursuer will travel at most $\frac{1-\nu^{n+1}}{\!\!\!1-\nu}\rho_0$ distance before capturing the evader. \hfill $\triangle$
\end{corollary}
\begin{proof}
    The proof follows directly from Proposition~\ref{prop:preResult} where we substitute the bounds on $\rcap$ to obtain the bounds on the number of sensing and the capture time.
\end{proof}

The results presented in this section assume that the purser is able to sense $n_{\max}$ times as well as the pursuer has sufficient fuel to complete the game. 
However, in reality, there might be constraints on the duration of the game\footnote{
Given the first order dynamics \eqref{eq:dynamics}, a constraint on the fuel is cast as a constraint on the game duration. 
} as well as on the maximum number of allowed sensing. 
For the rest of the paper, we focus on a problem where the game is played for a duration of $t_f$ and the maximum number of allowed sensing is $n$. 
The objective now is to analyze how the sensing strategy depends on $t_f$ and $n$. 
Note that the sensing strategies proposed in \cite{aleem2015self} and in Proposition~\ref{prop:preResult} are agnostic to $t_f$ and $n$.
In the subsequent sections, we formalize this problem and derive the optimal sensing strategy.
At this point, note that when $t_f \ge \frac{\rho_0 - \rcap}{1-\nu}$ and $n \ge \big\lfloor \frac{\log \rcap - \log \rho_0}{\log \nu} \big\rfloor$, Proposition~\ref{prop:preResult} is sufficient to construct a sensing strategy for the pursuer.

\section{Problem Formulation} \label{sec:formulation}

The pursuer's objective is to minimize its distance from the evader at the end of the game and, if possible, to capture the evader within the game duration, in which case the game ends as soon as the evader is captured.
The pay-off function for this game is as follows
\begin{align} \label{eq:payoff0}
    J = \begin{cases}
    0,   & \text{ if captured happened},\\
    \|\x_p(t_f)-\x_e(t_f)\| - \rcap, & \text{ otherwise}. 
    \end{cases}
\end{align}
The evader maximizes the pay-off while the pursuer minimizes it.
Although we are particularly interested in the pay-off function \eqref{eq:payoff0}, our analysis can be extended for a more general class of pay-off functions where 
\begin{align} \label{eq:payoff}
    J = \begin{cases}
    0, \qquad \qquad & \text{ if captured happened},\\
    \phi(\|\x_p(t_f)-\x_e(t_f)\|), & \text{ otherwise}. 
    \end{cases}
\end{align}
where we assume the following conditions on function $\phi$.

\begin{assumption}
\begin{itemize}
    \item[]  
    \item[(A1)] $\phi: \R \to \R_+$ is a non-decreasing function with $\phi(x) =~0$ for all $x\le \rcap$ and $\phi(x)>0$ otherwise.
    \item[(A2)] $\phi$ is convex.
\end{itemize}
\end{assumption}

Assumption (A1) incentivizes the pursuer to minimize its final distance from the evader. 
The convexity assumption keeps the formulation analytically tractable while encompassing a wider class of problems than the one presented in \eqref{eq:payoff0}.

Since the main objective of this work is to derive the optimal sensing strategy, we consider the motion strategy for the pursuer to be similar to what was given in \eqref{eq:purserV}.
\begin{assumption}
    \text{For all} $t \in (t_k, t_{k+1}]$, the pursuer follows 
    \[
   \v_p(t) = {\gamma(t)\r(t_k) }
    \]
     where $\gamma(t) \in [0,1]$ is to be designed by the pursuer. 
\end{assumption}

\begin{remark}
It is a rational choice to pick the pursuer's heading direction to be $\r(t_k)$ since the pursuer is forced to operate in `open-loop' in between sensing instances as it does not have access to $\r(t)$ continuously. 
Although it may seem that $\gamma(t)\equiv~1$ is the best choice, we observe from our analysis that, under certain choices for the parameters, $\gamma(t) = 0$ is optimal for some intervals of time. Note that $\gamma(t) =0$ represents a \textit{waiting} behavior in the pursuer's strategy. 
We prove that, under certain cases, waiting has a clear advantage over continuously moving and/or sensing earlier. 
The waiting behavior has proven to be a crucial characteristics in some of the recent works involving sensing  limited pursuit-evasion games \cite{shishika2021partial, pourghorban2022target, pourghorban2023target}. \hfill $\triangle$
\end{remark}


Although the dynamics \eqref{eq:dynamics} is deterministic, the outcome of the game is not necessarily deterministic if the players adopt randomized strategies.
Therefore, the pursuer/evader minimizes/maximizes $\E[J]$ where the expectation $\E[\cdot]$ is with respect to the randomized strategies taken by the pursuer and the evader. 
Let $\mu_e$ denote the motion strategy of the evader and $\mu_p$ denote the sensing strategy of the pursuer. 
The strategies are considered to be time dependent, however, to maintain notational brevity, we will suppress such time dependencies.
These strategies are measurable functions of the information sets of the players. 
To describe the information sets of the players, we first denote $m(t)$ to be the total number of sensing requests upto time $t$ and $\Tp(t) \triangleq \{t_0, t_1, \ldots, t_{m(t)}\}$ be the set of sensing instances upto time $t$, where $t_0$ is the initial time, which is assumed to be $0$ here, and $t_i< t_{i+1}$ for all $i$, and $t_{m(t)}$ is the latest sensing time.
For all $t$, we have $m(t) \le n$ and $t_{m(t)}\le t$.
Furthermore, we denote $\I_e(t) \triangleq \{ \x_e(s), \x_p(s), \Tp (t) ~|~ s\le t \}$ to be the information available to the evader at time $t$ and $\I_p(t) =\{\x_e(s'), \x_p(s), \Tp(t)~|~s'\in \Tp(t), s\le t\}$ to be the pursuer's available information. 

For a given pair of evader and pursuer strategies $(\mu_e,\mu_p)$, we define
$
\bar{J}(\mu_e,\mu_p) \triangleq \E[J]
$
to be the expected pay-off from the given strategy pair. 
In the subsequent sections, we derive an equilibrium pair ($\mu_e^*, \mu_p^*$) such that for any admissible strategies $\mu_e$ and $\mu_p$,
\begin{align*}
    \Bar{J}(\mu_e^*,\mu_p) \le \Bar{J}(\mu_e^*,\mu_p^*) \le \Bar{J}(\mu_e,\mu_p^*).
\end{align*}

\section{Analysis of the Game}
For the subsequent analysis, we define $V(\rho,\tau,\ell)$ to be the expected pay-off under a Nash equilibrium where $\rho$ is the current distance between the players, $\tau$ represents the remaining game duration, and $\ell$ is the number of remaining sensing requests.  
Notice that, given the current time $t$, we have $\tau = t_f -t$.
That is,
\begin{align*}
    V(\rho,\tau,\ell) = \inf_{\mu_p} \sup_{\mu_e}\E[J~\mid~&t=t_f - \tau,~ \|\x_e(t)-\x_p(t)\|= \rho, ~|\Tp(t)|=n-\ell],
\end{align*}
 where $J$ is defined in \eqref{eq:payoff} and $|\Tp(t)|$ denotes the cardinality of the set $\Tp(t)$. 
 Therefore, $V(\rho_0, t_f, n) = \Bar{J}(\mu_e^*,\mu_p^*)$, where recall that $\rho_0 = \|\x_p(0) - \x_e(0)\|$.

To compute $V(\rho,t_f,n)$, we first note that, for all $n$,
\begin{align} \label{eq:VBoundary}
    V(\rho, 0, n) =  \begin{cases}
        0,  & \text{if }  \rho \le \rcap,\\
        \phi(\rho), & \text{otherwise}.
    \end{cases} 
\end{align} 
Therefore, \eqref{eq:VBoundary} defines a boundary condition on $V$. 
To simplify our derivation for $V(\rho_0, t_f, n)$, let us divide the entire game duration into  intervals $[t_0, t_1), [t_1,t_2),\dots,[t_n, t_f)$, where $t_i$ denotes the $i$-th sensing time. 
We construct $V(\rho, \tau_i, n-\!i)$ backwards starting with $i=n$, where $\tau_i = t_f - t_i$. 
The following lemma computes $V(\rho,\tau,0)$ for any $\rho$ and $\tau$.

\begin{lemma} \label{lem:stage0}
    Let $\rho$ be the distance between the players at the last sensing instance. 
     Then, 
     \begin{align} \label{eq:lemma0}
        & V(\rho,\tau,0) \le \begin{cases}
            0, & \text{if } \tau \ge \rho,~ \nu\rho \le \rcap, \\
            \!\phi(\nu\tau + [\rho-\tau]^+\!), &\text{otherwise. } 
        \end{cases}
    \end{align}
    where we define $[x]^+ \triangleq \max\{x,0\}$. \hfill $\triangle$
\end{lemma}

\begin{proof}
    The proof is presented in Appendix~\ref{AP:lem:stage0}.
\end{proof}



Two important remarks are in order from Lemma~\ref{lem:stage0}, one regarding the tightness of the inequality \eqref{eq:lemma0} and the other regarding the optimality of a \textit{waiting behavior} in the pursuer's strategy. 

\begin{remark}
    Although Lemma~\ref{lem:stage0} provides an upper bound on $V(\rho, t_n, 0)$, the upper bound is tight everywhere except in the region $\Omega_0 \triangleq \{t_f-t_n\ge \rho, ~\rcap < \nu\rho \le \sqrt{1+\nu^2}\rcap$\} as discussed under \textbf{Case 2b} in Lemma~\ref{lem:stage0}. 
    In practice, this region is small as can be seen in Fig.~\ref{fig:upperBoundVisualization}(b).\hfill~$\triangle$
\end{remark}

\begin{remark}
    We notice in the proof of Lemma~\ref{lem:stage0} that, when $\tau \ge \rho$, the optimal motion strategy for the pursuer is to move to the last sensed location of the evader, $\x_e(t_n)$, and \textit{wait} there.  
    This is quite intuitive since moving at any arbitrary direction without the knowledge of the evader's location is harmful w.r.t. the expected pay-off. 
    The pursuer could have also waited at the beginning (or at any other point) of the game for a duration of $\tau-\rho$ and then move toward $\x_e(t_n)$.
    Alternatively, the pursuer could have moved for the entire duration with a slower speed. 
    In summary, all the scenarios lead to the same conclusion that $\gamma(t) \equiv 1$ is not an optimal strategy for the entire duration. \hfill~$\triangle$
\end{remark}

\begin{figure}
    \centering
    \begin{subfigure}[b]{0.45 \linewidth}
    \includegraphics[trim = 10 0 25 10, clip, width =  \linewidth]{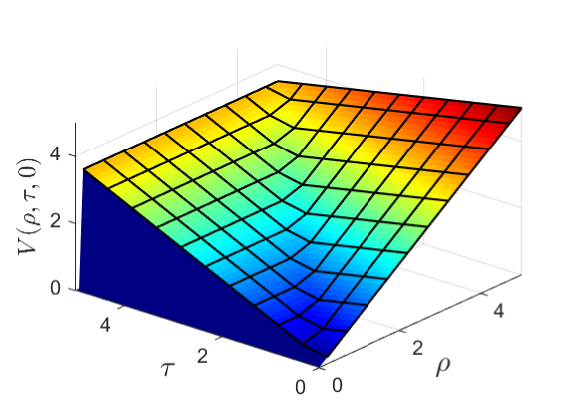}
    \caption{}
    \end{subfigure}
     \begin{subfigure}[b]{0.45 \linewidth}
         \includegraphics[trim = 10 0 25 10, clip, width =  \linewidth]{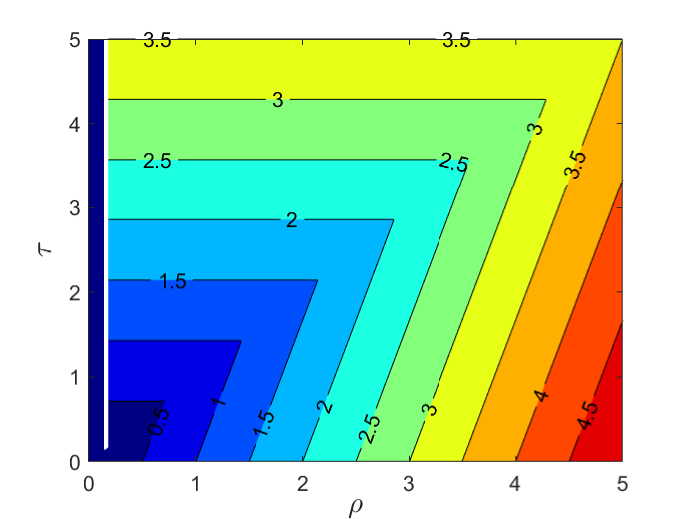}
         \caption{}
     \end{subfigure}
    \caption{\small We choose $\rcap = 0.1$ and $\nu = 0.7$, and $\phi(x)= [x - \rcap]^+$ for this plot. (a) The upperbound of $V(\rho,\tau,0)$ from Lemma~\ref{lem:stage0}. (b) Contour plot of $V(\rho,\tau,0)$ on the $(\rho,\tau)$ plane. The numbers on each contour show the value of $V(\rho,\tau,0)$ along that contour. The thin vertical white slice approximately along $\rho = 0.15$ illustrates the region $\Omega_0$.
    This is the only region in which we have an upper bound of $V(\rho,\tau,0)$ instead of an exact expression.}
    \label{fig:upperBoundVisualization}
    \vspace{- 6 pt}
\end{figure}

The following theorem is the main result of this section where we provide an upper bound on $V(\rho,\tau,\ell)$ for all $\rho, \tau$ and $\ell \le n$. 
The upper bound is tight everywhere except in the region $\{\tau \ge \frac{1-\nu^{\ell +1}}{1-\nu}\rho, \rcap \le \nu \rho \le \sqrt{1+\nu^2}\rcap\} \subseteq \Omega_0$.

\begin{theorem} \label{thm:main}
 Let $t_k$ be the $k$-th sensing time and $\rho$ be the distance between the players at that instance. 
    Let $\tau = t_f - t_k$ and $ \ell = n -k$ denote the remaining time and the remaining number of sensing, respectively. Then, 
 \begin{align*}
     V(\rho, \tau, \ell) \le  \begin{cases}
         0, \qquad \quad\quad\text{if } \tau \ge\!\!&\frac{1-\nu^{\ell +1}} {1-\nu}\rho, ~~\nu^{\ell +1}\rho \le \rcap,\\
         \phi(\nu\tau+\rho - \tau), &\text{if } \tau \le \frac{1-\nu^{\ell +1}} {1-\nu}\rho,\\
         \phi(\frac{1-\nu}{1-\nu^{\ell +1}}\nu^{\ell +1}\tau), & \text{otherwise}. \hfill \triangle
     \end{cases}
 \end{align*} 

\end{theorem}

\begin{proof} (Sketch)
 The proof is presented in Appendix~\ref{AP:thm:main}. 
\end{proof}

\begin{remark}\label{rem:waiting}
    $V(\rho_0, t_f, n)$ is a special case of Theorem~\ref{thm:main}.
    According to Theorem~\ref{thm:main}, the pursuer must wait before requesting the first sensing. The waiting time is constructed in such a way that the total distance traveled starting from $t_1$ (i.e., the first sensing instance) is the same as the duration left. 
    In this way the pursuer does not need to wait in the future. 
    Furthermore, it can be shown that, if the pursuer requests the first sensing before waiting for $\frac{1-\nu}{1-\nu^{\ell+1}}\tau - \rho $, then it results in a suboptimal outcome for the pursuer. The proof is provided in Appendix~\ref{AP:waiting}. \hfill $\triangle$
\end{remark}

Using the result from Theorem~\ref{thm:main}, we plot $V(\rho,\tau,\ell)$ in Fig.~\ref{fig:Vn}. 
For a given $\rho$ and $\tau$, $V(\rho,\tau,\ell)$ is a non-increasing function of $\ell$, as one would have expected. 
In the next section, we investigate how $V(\rho_0,t_f, n)$ behaves with respect $n$.

\begin{figure*}
\includegraphics[trim = 0 0 10 10, clip, width = 0.19 \linewidth]{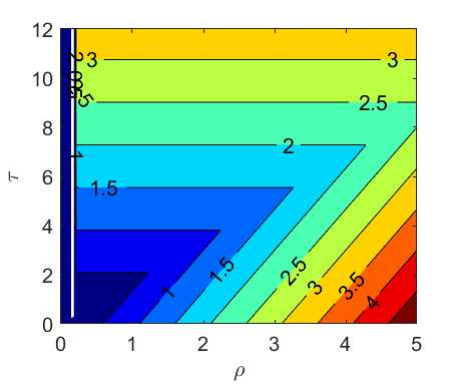}
\includegraphics[trim = 0 0 10 10, clip, width = 0.19 \linewidth]{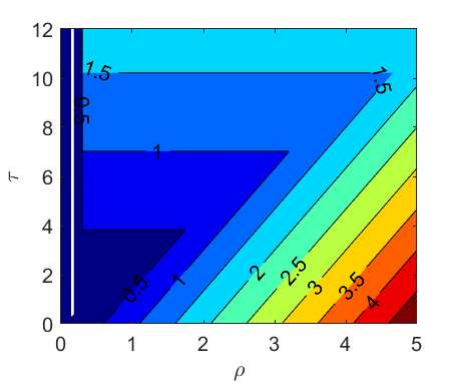}
\includegraphics[trim = 0 0 10 10, clip, width = 0.19 \linewidth]{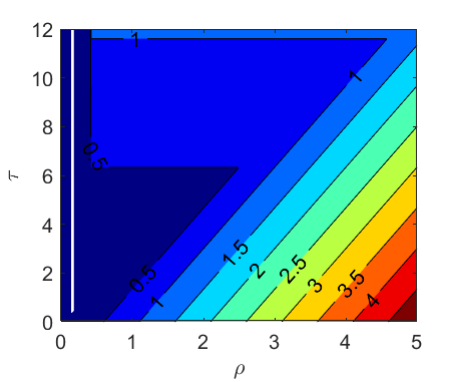}
\includegraphics[trim = 0 0 10 10, clip, width = 0.19 \linewidth]{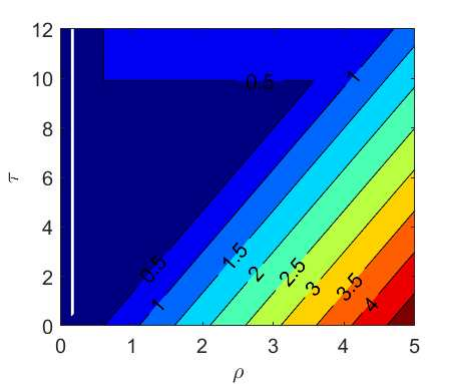}
\includegraphics[trim = 0 0 10 10, clip, width = 0.19 \linewidth]{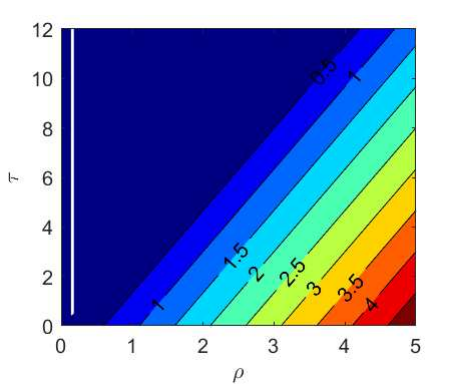}
\caption{\small Contour plots of $V(\rho,\tau,\ell)$ from $\ell=1$ to $\ell = 5$. $\ell$ gradually increases from the left most subfigure to the right most one.}
\label{fig:Vn}
\vspace{- 6 pt}
\end{figure*}

\section{Performance degradation due to sensing limitation} \label{sec:senseN}
The objective of this section is to quantify the degradation in the pursuer's performance due to its sensing limitation. 
For an initial distance of $\rho_0$ and a game duration of $t_f$, the pursuer's pay-off under continuous sensing is 
$\phi([\rho_0 - (1-\nu)t_f]^+)$, which will serve as the baseline for quantifying the performance degradation.
In particular, for a given $\nu,\rho_0$ and $t_f$, we use the metric $\delta(n) \triangleq V(\rho_0,0,n) - \phi([\rho_0 - (1-\nu)t_f]^+)$ to quantify the performance degradation.
Since we only have an upper bound of $V(\rho,\cdot,\cdot)$ in $\Omega_0$, we exclude this region from the discussion in this section and make the following assumption.
\begin{assumption}
    We assume that $\nu \rho_0 > \sqrt{1 + \nu^2} \rcap$.
\end{assumption}

In the following, we first state a lemma that provides the required number of sensing $n^*$ to ensure $\delta(n^*)=0$.  
\begin{lemma} \label{lem:mSensing}
    For a given initial distance $\rho_0$ and a game duration $t_f$, the pursuer's pay-off is $\phi([\rho_0 - (1-\nu)t_f]^+)$ with $n^*$ sensing, where
    \begin{align} \label{eq:mSensing}
        n^* = \begin{cases}
            \big\lfloor\frac{\log(\rho_0 - (1-\nu)t_f) - \log(\rho_0)}{\log (\nu)}\big\rfloor, &\text{if }  t_f < \frac{\rho_0 - \rcap}{1-\nu},\\
            \big\lfloor\frac{\log(\rcap) - \log(\rho_0)}{\log (\nu)}\big\rfloor, &\text{otherwise}.
        \end{cases}
    \end{align} 
    \hfill $\triangle$
\end{lemma}
\begin{proof}
    The proof is presented in Appendix~\ref{AP:Lemma3}.
\end{proof}

Lemma~\ref{lem:mSensing} is an extension of Proposition~\ref{prop:preResult} where we now have incorporated the role of  $t_f$ into the number of required sensing. 
Proposition~\ref{prop:preResult} was derived under the case in which the pursuer had sufficient time (i.e., $t_f \ge \frac{\rho_0 - \rcap}{1-\nu}$) to capture.

We now derive the degradation in pursuer's performance when the available number of sensing is less than $n^*$. 
We discuss it separately for $t_f \ge \frac{\rho_0-\rcap}{1-\nu}$ and for $t_f < \frac{\rho_0-\rcap}{1-\nu}$. 
When $t_f \ge \frac{\rho_0-\rcap}{1-\nu}$, the purser is able to capture the evader according to Lemma~\ref{lem:mSensing} whereas with $n$ number of sensings the pursuer's pay-off is $\phi(\frac{1-\nu}{1-\nu^{n+1}}\nu^{n+1}t_f)$  and therefore, $\delta(n) = \phi(\frac{1-\nu}{1-\nu^{n+1}}\nu^{n+1}t_f)$. 

On the other hand, when  $t_f\!<\!\frac{\rho_0-\rcap}{1-\nu}$, one may verify that
\[
t_f \ge \frac{\rho_0 - \nu^{n+1}\rho_0}{1-\nu}, \quad \text{and  }  \nu^{n+1}\rho_0 > \rcap 
\]
 for any integer $n<n^*$, where $n^*$ is given in \eqref{eq:mSensing}. 
Therefore, for all $n<n^*$, we obtain
\begin{align*}
    \delta(n) &= \phi(\frac{1-\nu}{1-\nu^{n+1}}\nu^{n+1}t_f) - \phi(\rho_0 - (1-\nu)t_f) \\
    & \ge \underset{\triangleq \beta(n)}{\underbrace{  \big(\frac{\nu^{n+1}}{1-\nu^{n+1}}\frac{(1-\nu)t_f}{\rho_0 -(1-\nu)t_f} - 1 \big)}}\phi(\rho_0 - (1-\nu)t_f)
\end{align*}
where we have used the Jensen's inequality for the convex function $\phi$.\footnote{For $x\ge y>0$, $\phi(y) = \phi(\lambda x) \le \lambda \phi(x) + (1-\lambda) \phi(0) = \lambda \phi(x)$, where $\lambda = \frac{y}{x}$.}
Notice that the coefficient $\beta(n)$ depends on $n$ exponentially and therefore, the degradation is exponential with $n$.
This is demonstrated in Fig.~\ref{fig:beta(n)}. 

\begin{figure}
    \centering
    \includegraphics[trim = 10 0 10 10, clip, width = 0.6 \linewidth]{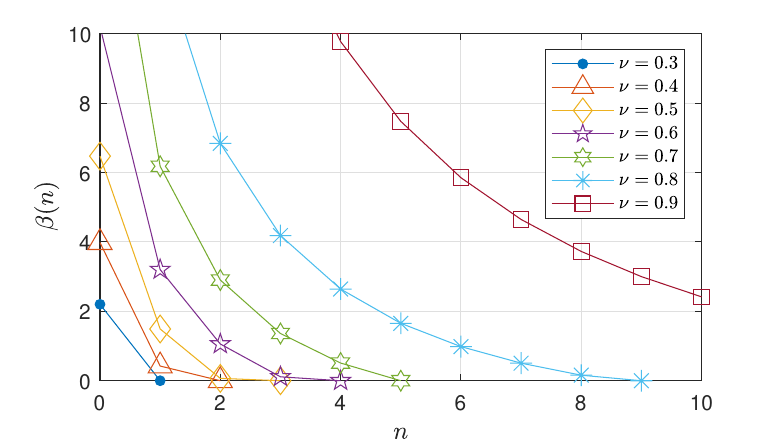}
    \caption{$\beta(n)$ vs. $n$ for different values of $\nu$. For all the plots we pick $\rcap = 0.1,$ $ \rho_0 = 5$ and $t_f = 0.9\frac{\rho_0 - \rcap}{1-\nu}$.}
    \label{fig:beta(n)}
\end{figure}


\section{Conclusions} \label{sec:conslusions}
In this work, we considered a sensing limited pursuit-evasion game where the pursuer is restricted to intermittent sensing.
We derived the number of required sensing $n_{\max}$ in Proposition~\ref{prop:preResult} to ensure capture.
Next, we considered the game under explicit sensing and time/fuel budgets, i.e., the pursuer has $t_f$ amount of time to capture the evader while using a maximum of $n$ sensing requests. 
An upper bound on the value function for this game has been obtained, where the upper bound is tight everywhere except in a very small region $\Omega_0$. 
Future work will focus on the case where the evader is also equipped with an intermittent sensing capability. 
Finding an equilibrium strategy is necessary for applications where sensing is expensive and/or undesired for both agents. 
In addition, tightening the upper bound of $V$ in $\Omega_0$ is also an open problem.

\bibliographystyle{IEEEtran}
\bibliography{arxiv.bib}

\begin{thebibliography}{10}
\providecommand{\url}[1]{#1}
\csname url@samestyle\endcsname
\providecommand{\newblock}{\relax}
\providecommand{\bibinfo}[2]{#2}
\providecommand{\BIBentrySTDinterwordspacing}{\spaceskip=0pt\relax}
\providecommand{\BIBentryALTinterwordstretchfactor}{4}
\providecommand{\BIBentryALTinterwordspacing}{\spaceskip=\fontdimen2\font plus
\BIBentryALTinterwordstretchfactor\fontdimen3\font minus
  \fontdimen4\font\relax}
\providecommand{\BIBforeignlanguage}[2]{{%
\expandafter\ifx\csname l@#1\endcsname\relax
\typeout{** WARNING: IEEEtran.bst: No hyphenation pattern has been}%
\typeout{** loaded for the language `#1'. Using the pattern for}%
\typeout{** the default language instead.}%
\else
\language=\csname l@#1\endcsname
\fi
#2}}
\providecommand{\BIBdecl}{\relax}
\BIBdecl

\bibitem{isaacs1999differential}
R.~Isaacs, \emph{Differential games: a mathematical theory with applications to
  warfare and pursuit, control and optimization}.\hskip 1em plus 0.5em minus
  0.4em\relax Courier Corporation, 1999.

\bibitem{yan2016multi}
C.~Yan and T.~Zhang, ``Multi-robot patrol: A distributed algorithm based on
  expected idleness,'' \emph{International Journal of Advanced Robotic
  Systems}, vol.~13, no.~6, 2016.

\bibitem{robotics9020047}
T.~Alam and L.~Bobadilla, ``Multi-robot coverage and persistent monitoring in
  sensing-constrained environments,'' \emph{Robotics}, vol.~9, no.~2, 2020.

\bibitem{inproceedings}
E.~García, A.~Von~Moll, D.~Casbeer, and M.~Pachter, ``Strategies for defending
  a coastline against multiple attackers,'' in \emph{58th Conference on
  Decision and Control}.\hskip 1em plus 0.5em minus 0.4em\relax IEEE, 2019, pp.
  7319--7324.

\bibitem{sun2017multiple}
W.~Sun, P.~Tsiotras, T.~Lolla, D.~N. Subramani, and P.~F. Lermusiaux,
  ``Multiple-pursuer/one-evader pursuit--evasion game in dynamic flowfields,''
  \emph{Journal of guidance, control, and dynamics}, vol.~40, no.~7, pp.
  1627--1637, 2017.

\bibitem{oyler2016pursuit}
D.~W. Oyler, P.~T. Kabamba, and A.~R. Girard, ``Pursuit--evasion games in the
  presence of obstacles,'' \emph{Automatica}, vol.~65, pp. 1--11, 2016.

\bibitem{bhattacharya2010existence}
S.~Bhattacharya and S.~Hutchinson, ``On the existence of {N}ash equilibrium for
  a two-player pursuit—evasion game with visibility constraints,'' \emph{The
  International Journal of Robotics Research}, vol.~29, no.~7, pp. 831--839,
  2010.

\bibitem{aleem2015self}
S.~A. Aleem, C.~Nowzari, and G.~J. Pappas, ``Self-triggered pursuit of a single
  evader,'' in \emph{54th Conference on Decision and Control}.\hskip 1em plus
  0.5em minus 0.4em\relax IEEE, 2015, pp. 1433--1440.

\bibitem{aleem2015bself}
------, ``Self-triggered pursuit of a single evader with uncertain
  information,'' \emph{arXiv preprint arXiv:1512.06184}, 2015.

\bibitem{maity2016strategies}
D.~Maity and J.~S. Baras, ``Strategies for two-player differential games with
  costly information,'' in \emph{13th International Workshop on Discrete Event
  Systems}.\hskip 1em plus 0.5em minus 0.4em\relax IEEE, 2016, pp. 211--216.

\bibitem{maity2016optimal}
------, ``Optimal strategies for stochastic linear quadratic differential games
  with costly information,'' in \emph{55th Conference on Decision and Control
  (CDC)}.\hskip 1em plus 0.5em minus 0.4em\relax IEEE, 2016, pp. 276--282.

\bibitem{bopardikar2007cooperative}
S.~D. Bopardikar, F.~Bullo, and J.~P. Hespanha, ``Cooperative pursuit with
  sensing limitations,'' in \emph{American Control Conference}.\hskip 1em plus
  0.5em minus 0.4em\relax IEEE, 2007, pp. 5394--5399.

\bibitem{durham2010distributed}
J.~W. Durham, A.~Franchi, and F.~Bullo, ``Distributed pursuit-evasion with
  limited-visibility sensors via frontier-based exploration,'' in
  \emph{International Conference on Robotics and Automation}.\hskip 1em plus
  0.5em minus 0.4em\relax IEEE, 2010, pp. 3562--3568.

\bibitem{shishika2021partial}
D.~Shishika, D.~Maity, and M.~Dorothy, ``Partial information target defense
  game,'' in \emph{International Conference on Robotics and Automation}.\hskip
  1em plus 0.5em minus 0.4em\relax IEEE, 2021, pp. 8111--8117.

\bibitem{gerkey2006visibility}
B.~P. Gerkey, S.~Thrun, and G.~Gordon, ``Visibility-based pursuit-evasion with
  limited field of view,'' \emph{The International Journal of Robotics
  Research}, vol.~25, no.~4, pp. 299--315, 2006.

\bibitem{maity2017linear}
D.~Maity, A.~Anastasopoulos, and J.~S. Baras, ``Linear quadratic games with
  costly measurements,'' in \emph{56th Conference on Decision and
  Control}.\hskip 1em plus 0.5em minus 0.4em\relax IEEE, 2017, pp. 6223--6228.

\bibitem{maity2017asymptotic}
D.~Maity and J.~S. Baras, ``Asymptotic policies for stochastic differential
  linear quadratic games with intermittent state feedback,'' in \emph{25th
  Mediterranean Conference on Control and Automation}.\hskip 1em plus 0.5em
  minus 0.4em\relax IEEE, 2017, pp. 117--122.

\bibitem{huang2021defending}
Y.~Huang, J.~Chen, and Q.~Zhu, ``Defending an asset with partial information
  and selected observations: A differential game framework,'' in \emph{60th
  Conference on Decision and Control}.\hskip 1em plus 0.5em minus 0.4em\relax
  IEEE, 2021, pp. 2366--2373.

\bibitem{pourghorban2022target}
A.~Pourghorban, M.~Dorothy, D.~Shishika, A.~Von~Moll, and D.~Maity, ``Target
  defense against sequentially arriving intruders,'' in \emph{61st Conference
  on Decision and Control}.\hskip 1em plus 0.5em minus 0.4em\relax IEEE, 2022,
  pp. 6594--6601.

\bibitem{pourghorban2023target}
A.~Pourghorban and D.~Maity, ``Target defense against periodically arriving
  intruders,'' \emph{arXiv preprint arXiv:2303.05577}, 2023.

\end{thebibliography}

\appendix 

\section{Existence of a Surviving Strategy for the Evader when $\nu \rho > \rcap$} \label{AP:evadingStrategy}
To see this, let the evader follow the direction $R(\theta)\r(t_n)$ with $\theta = \cos^{-1}\nu$, where $R(\theta) = \begin{bmatrix}
    \cos(\theta) & -\sin(\theta)\\
    \sin(\theta) & \cos(\theta)
\end{bmatrix} $ is a rotation matrix. 
In this case, 
Notice that, 
\begin{align*}
    \x_p(t) &= \x_p(t_n) + \r(t_n) \int_{s=t_n}^t\gamma(s) ds,\\
    \x_e(t) &= \x_e(t_n) +  \nu(t-t_n) R(\cos^{-1}\nu )\r(t_n),
\end{align*}
where $\gamma(s) = 1$ for all $s\in [t_n, t_n+\rho]$ and zero otherwise. 
Therefore, for all $t\le t_n + \rho$
\begin{align*}
    \|\x_e(t)-\x_p(t)\| & = \| \x_e(t_n) - \x_p(t_n)   - \r(t_n) (t-t_n) + \nu (t-t_n) R(\cos^{-1}\nu) \r(t_n) \| \\
    & = \| \rho \r(t_n) - (t - t_n)\r(t_n) + \nu (t-t_n) R(\cos^{-1}\nu) \r(t_n) \| \qquad [\text{since } \|\x_e(t_n) -\x_p(t_n)\| = \rho] \\
    & = \sqrt{(\rho - (t-t_n))^2 + \nu^2 (t-t_n)^2 + 2\nu^2 (t-t_n)(\rho - (t - t_n))} \\
    & = \sqrt{\nu^2\rho^2 + (1-\nu^2)(\rho - (t-t_n))^2} \ge \nu\rho > \rcap.
\end{align*}
That is, capture does not happen during the interval $[t_n, t_n + \rho]$. 
Now for any $t> t_n + \rho$,
\begin{align*}
    \x_p(t) &= \x_p(t_n) + \r(t_n) \int_{s=t_n}^{t_n + \rho}\gamma(s) ds = \x_p(t_n) + \rho \r(t_n) = \x_e(t_n),\\
    \x_e(t) &= \x_e(t_n) +  \nu(t-t_n) R(\cos^{-1}\nu )\r(t_n).
\end{align*}
Therefore,
\begin{align*}
    \|\x_e(t)-\x_p(t)\| & = \nu (t - t_n) > \nu \rho > \rcap,
\end{align*}
where the first inequality is obtained by using $t> t_n + \rho$. This proves the claim that there exists a trajectory for the evader such that capture does not happen.

\section{Proof of Lemma~\ref{lem:stage0}} \label{AP:lem:stage0}

 ~~~~~\textbf{{Case 1: $\tau \ge \rho,~  \nu \rho \le \rcap$}} \\
In this case, the pursuer moves along $\r(t_n)$ and capture is inevitable regardless of the evader's strategy.\\

\textbf{{Case 2a: $\tau \ge \rho,~  \nu\rho > \sqrt{1+\nu^2}\rcap$}} \\
In this case, the pursuer shall move to the last sensed location of the evader and stay there until the end of the game. 
The evader shall randomly pick either $\r(t_n)^\perp$ or $-\r(t_n)^\perp$ with equal probability and move along that direction with maximum speed, where $\r(t_n)^\perp$ denotes an unit vector perpendicular to $\r(t_n)$. 
Notice that, 
\begin{align*}
    \x_p(t) &= \x_p(t_n) + \r(t_n) \int_{s=t_n}^t\gamma(s) ds,\\
    \x_e(t) &= \x_e(t_n) + \theta \nu(t-t_n) \r(t_n)^\perp,
\end{align*}
where $\gamma(s) = 1$ for all $s\in [t_n, t_n+\rho]$ and zero otherwise, and $\theta \in \{-1, 1\}$ is chosen uniform randomly by the evader at time $t_n$. 
One may verify that the minimum distance between the players during the interval $[t_n, t_n+\rho]$ is $\frac{\nu \rho}{\sqrt{1+\nu^2}}$. 
Since we are considering the case $\nu\rho > \sqrt{1+\nu^2} \rcap$, capture does not happen during the interval $[t_n, t_n + \rho]$.
To show that it is an equilibrium pair we prove that unilateral deviation in any of the player's strategy does not improve their pay-off. 
From the evader's perspective, since the pursuer is going to stop at $\x_e(t_n)$, the final distance between the players is independent of the evader's heading direction as long as the evader does not get captured before time $t_n+\rho$. 
Therefore, the evader cannot receive a better pay-off by deviating from this strategy. 
On the other hand, let us assume that the pursuer deviates from the proposed strategy and ends up at a different point $\x_p(t_n) + \alpha_1 \r(t_n) + \alpha_2 \r(t_n)^\perp$ where $\alpha_1, \alpha_2 \in [-\tau, \tau]$ and $ \sqrt{\alpha_1^2 + \alpha_2^2} \le \tau$.
One may verify that, 
\begin{align*}
    \|\x_e(t_f) - \x_p(t_f)\| &= \|(\rho - \alpha_1) \r(t_n) + (\theta \nu \tau-\alpha_2) \r(t_n)^\perp\|\\
    & = \sqrt{(\rho-\alpha_1)^2 + ( \theta \nu\tau-\alpha_2)^2} \triangleq g(\theta).
\end{align*}
Notice that $g(\theta)$ is convex in $\theta$, and therefore, taking expectation of $g(\theta)$ and using Jensen's inequality yields 
\begin{align}
    \E[g(\theta)] &\ge \sqrt{(\rho-\alpha_1)^2 + \E[( \theta \nu\tau-\alpha_2)^2]} \nonumber\\
& = \sqrt{(\rho-\alpha_1)^2 + \nu^2\tau^2 + \alpha_2^2 }, \label{eq:Expected_g}
\end{align}
where the last equality is obtained using $\theta^2 = 1$ and $\E[\theta]= 0$.
Therefore, the expected pay-off is 
\begin{align*}
    \E[\phi(\|\x_e(t_f) - \x_p(t_f)\|)] & \overset{(\dagger)}{\ge}  \phi(\E[\|\x_e(t_f) - \x_p(t_f)\|]) \\
    & \overset{(\ddagger)}{\ge} \phi(\sqrt{(p-\alpha_1)^2 +\nu^2\tau^2 + \alpha_2^2}),
\end{align*}
where the inequality $(\dagger)$ is obtained by using Jensen's inequality to the convex function $\phi$. 
The second inequality $(\ddagger)$ is obtained by combining the non-decreasing nature of $\phi$ (Assumption~1) and \eqref{eq:Expected_g}.
Since $\phi$ is non-decreasing, we may further conclude that the optimal choices for $\alpha_1$ and $\alpha_2$ are $\rho$ and $0$, respectively. 
Consequently, the expected pay-off is $\phi(\nu \tau)$, which is the same as $\phi(\nu \tau + [\rho -\tau]^+)$
since $\tau \ge \rho$.

\textbf{{Case 2b: $\tau \ge \rho,~  \rcap < \nu\rho \le \sqrt{1+\nu^2}\rcap$}}\\
We prescribe the pursuer to follow the same strategy as in the last case. 
Notice that, since $\nu\rho \le (1+\nu^2)\rcap$, the evader cannot move in the perpendicular directions for the entire duration without being captured.
Nonetheless, given that the evader picks a trajectory that does not lead to capture (such a trajectory exists\footnote{For example, let the evader pick the direction $R(\theta) \r(t_n)$ with $\theta = \cos^{-1}\nu$ where $R(\theta) \triangleq \begin{bmatrix}
    \cos(\theta) & -\sin(\theta)\\
    \sin(\theta) & \cos(\theta)
\end{bmatrix}$.
A detailed proof is available in Appendix~\ref{AP:evadingStrategy}.}), the final distance between the players will be at most $\nu\tau + (\rho-\tau)$. 
Consequently, the pay-off from this case is no more than $\phi(\nu\tau+[\rho-\tau]^+)$.\\

\textbf{{Case 3 $\tau < \rho$}:} \\
In this case, the optimal strategy for both the players is to move along $\r(t_n)$ for the entire duration. 
One may verify that the proposed strategies constitute an equilibrium pair by considering unilateral deviations in the players' strategies.  
In this case, 
$\|\x_e(t_f) - \x_p(t_f)\|= \rho -(1-\nu)\tau = \nu \tau + (\rho -\tau)$ and thus, the pay-off can be written as $\phi(\nu\tau + [\rho-\tau]^+)$.

 This completes the proof.

\section{Proof of Theorem~\ref{thm:main}} \label{AP:thm:main}

 The proof follows similar steps as Lemma~\ref{lem:stage0}. 

\textbf{Case 1: $\tau \ge \frac{1-\nu^{\ell +1}} {1-\nu}\rho, ~~\nu^{\ell +1}\rho \le \rcap$} \\
The strategy for the pursuer is to go to the last sensed location of the evader and request for sensing. 
The evader on the other hand shall go along one of the directions $\pm \r(t_i)^\perp$ chosen uniformly randomly for the interval $(t_i, t_{i+1}]$, where $i=(n-\ell), (n-\ell +1),\ldots, n$, and we define $t_{n+1} \triangleq t_f$.

\textbf{Case 2: $\tau \le \frac{1-\nu^{\ell +1}} {1-\nu}\rho$}\\
The strategy for the pursuer remains the same as the last case. 
The strategy for the evader also remains the same for the intervals $\{(t_i, t_{i+1}]\}_{i=n-\ell}^{n -1}$.  
During the interval $(t_n, t_f]$, the evader shall go along $\r(t_n)$ (this interval falls under \textbf{Case~3} of Lemma~\ref{lem:stage0}).

\textbf{Case 3: $\tau \ge \frac{1-\nu^{\ell +1}} {1-\nu}\rho, ~~\nu^{\ell +1}\rho > \rcap$}\\
The strategy for the pursuer is to go to $\x_e(t_{n-\ell})$ and wait there for $w = \frac{1-\nu}{1-\nu^{\ell+1}}\tau - \rho $ amount of time before requesting the first sensing.
Therefore, at the moment of sensing, the remaining distance, time, sensing requests are  $\rho' = \nu \frac{1-\nu}{1-\nu^{\ell+1}}\tau$, $\tau' = \frac{\nu(1 - \nu^\ell)}{1-\nu^{\ell+1}}\tau $, and $\ell' = \ell -1$, respectively.
Since $\rho',\tau'$, and $\ell'$ satisfy $\tau' = \frac{1 - \nu^{\ell'+1}}{1-\nu}\rho'$, we may invoke \textbf{Case 2} of this theorem to conclude
 $V(\rho', \tau',\ell') = \phi(\nu\tau' + \rho' - \tau')$.  
By substituting the expressions of $\rho',\tau'$ and $\ell'$, we obtain 
$
\phi(\nu\tau' + \rho' - \tau') = \phi(\frac{1-\nu}{1-\nu^{\ell +1}}\nu^{\ell +1}\tau).
$
If $\sqrt{1 +\nu^2}\rcap< \nu\rho$, the evader must move along $\pm \r(t_i)^\perp$ (chosen uniformly randomly) for $i=(n\!-\!\ell),\dots,  (n\!-\!1)$, and consequently $ V(\rho,\tau,\ell) = \phi(\frac{1-\nu}{1-\nu^{\ell +1}}\nu^{\ell +1}\tau)$.
If $\sqrt{1 +\nu^2}\rcap \ge \nu\rho$, the evader must move in a fashion to escape from being captured in the first interval $(t_{n-\ell}, t_{n- \ell +1}]$. In this case, $\phi(\frac{1-\nu}{1-\nu^{\ell +1}}\nu^{\ell +1}\tau)$ is an upper bound of $V(\rho,\tau,\ell)$.

\section{Proof of Remark~\ref{rem:waiting}}\label{AP:waiting}
To prove this remark we denote $\Delta_i$ to be the wait time before the $i$-th sensing. 
That is, given the distance $\rho_{i-1}$ between the players at the moment of the $(i-1)$-th sensing, the pursuer spends $\rho_{i-1}$ amount of time to move to the last sensed location of the evader and then waits for $\Delta_i$ amount of time before it requests for the $i$-th sensing. 
Meanwhile, the evader moves along a straight line.
Therefore, at the $i$-th sensing moment the distance between the players will be
\begin{align} \label{eq:rhoi}
\rho_i = \nu(\rho_{i-1}+\Delta_i) = \nu^i \rho_0 + \sum_{k=1}^i \nu^{i+1 - k} \Delta_k, 
\end{align}
where $\rho_0$ is the initial distance. 
At the last (i.e., $n$-th) sensing moment, the distance between them is $\rho_n=\nu^{n}\rho_0 + \nu^{n}\Delta_1 + \nu^{n-1} \Delta_2 + \cdots + \nu\Delta_n$ and the total time elapsed is $(\rho_0 + \Delta_1 + \rho_1 + \Delta_2+\cdots +\rho_{n-1}+\Delta_n)$. 
Using \eqref{eq:rhoi} we may simplify
\begin{align*}
    (\rho_0 + \Delta_1 + \rho_1 + \Delta_2+\cdots +\rho_{n-1}+\Delta_n) = \frac{1-\nu^{n}}{1-\nu}\rho_0 + \frac{1-\nu^{n}}{1-\nu}\Delta_1  + \frac{1-\nu^{n-1}}{1-\nu}\Delta_2  + \cdots +\Delta_n. 
\end{align*}
Therefore, the remaining game duration is $\tau \triangleq t_f - \Big( \frac{1-\nu^{n}}{1-\nu}\rho_0 + \frac{1-\nu^{n}}{1-\nu}\Delta_1  + \frac{1-\nu^{n-1}}{1-\nu}\Delta_2  + \cdots +\Delta_n \Big)$.
By invoking Lemma~\ref{lem:stage0}, we obtain 
 \begin{align} \label{eq:V_appendix}
        & V(\rho_n,\tau,0) = \begin{cases}
            0, & \text{if } \tau \ge \rho_n,~ \nu\rho_n \le \rcap, \\
            \!\phi(\nu\tau + [\rho_n-\tau]^+\!), &\text{otherwise. } 
        \end{cases}
    \end{align}
At this point the pursuer shall choose $\{\Delta_i\}_{i=1}^n$ such that $V(\rho_n,\tau, 0)$ is minimized. 
To simplify the problem, we consider the following three parametric cases separately: (i) $t_f < \frac{1 -\nu^{n+1}}{1-\nu}\rho_0$, (ii) $t_f \ge \frac{1 -\nu^{n+1}}{1-\nu}\rho_0$ and $\nu^{n+1}\rho_0 \le \rcap$, and (iii) $t_f \ge \frac{1 -\nu^{n+1}}{1-\nu}\rho_0$ and $\nu^{n+1}\rho_0 > \rcap$.\\

\textbf{Case (i):} In this case, for any $\{\Delta_i\}_{i=1}^n$ such that $\Delta_i\ge 0$ for all $i$, we have
\begin{align}
    \rho_n - \tau = \frac{1-\nu^{n+1}}{1-\nu}\rho_0 - t_f +  \Big( \frac{1-\nu^{n+1}}{1-\nu}\Delta_1  + \frac{1-\nu^{n}}{1-\nu}\Delta_2  + \cdots + \frac{1-\nu^2}{1-\nu}\Delta_n \Big) > 0,
\end{align}
i.e., $\tau<\rho_n$. 
Therefore, from \eqref{eq:V_appendix}, we have 
\begin{align*}
    V(\rho_n,\tau,0) = \phi(\nu\tau + [\rho_n-\tau]^+\!) = \phi(\rho_0 - t_f+\sum_{i=1}^n\Delta_n),
\end{align*}
which proves that the optimal $\Delta_i$ must be $0$ for all $i$.\\

\textbf{Case (ii):} In this case, let us choose $\Delta_i = 0$ for all $i$. 
Therefore, $\tau = t_f - \frac{1-\nu^n}{1-\nu}\rho_0 = \nu^n\rho_0 + \Big( t_f  - \frac{1 -\nu^{n+1}}{1-\nu}\rho_0 \Big) \ge \nu^n\rho_0 = \rho_n$ and $\nu\rho_n = \nu^{n+1}\rho_0 \le \rcap$. 
Thus, the first condition in \eqref{eq:V_appendix} is satisfied and we obtain $V(\rho_n,\tau,0) = 0$, which is the least possible value of $V$. 
Thus, $\Delta_i = 0$ for all $i$ is an optimal choice.\\

\textbf{Case (iii):} In this case we have $\tau\ge \rho_n$ and $\nu\rho_n > \rcap$ and therefore, from \eqref{eq:V_appendix},
\begin{align*}
    V(\rho_n,\tau,0) = \phi(\nu\tau + [\rho_n-\tau]^+\!).
\end{align*}

In order to minimize $V(\rho_n,\tau,0)$, we may alternatively consider the following optimization problem:
\begin{align} \label{eq:optimization}
\begin{split}
    \min_{\Delta_1,\Delta_2,\dots,\Delta_n, \rho_n,\tau} \qquad \qquad & \nu\tau + [\rho_n-\tau]^+ \\
    \text{subject to  } & \tau = t_f - \Big( \frac{1-\nu^{n}}{1-\nu}\rho_0 + \frac{1-\nu^{n}}{1-\nu}\Delta_1  + \frac{1-\nu^{n-1}}{1-\nu}\Delta_2  + \cdots +\Delta_n \Big),\\
    & \rho_n = \nu^{n}\rho_0 + \nu^{n}\Delta_1 + \nu^{n-1} \Delta_2 + \cdots + \nu\Delta_n, \\
    & \Delta_i \ge 0, \qquad \forall i.
\end{split}
\end{align}

First, we prove (by contradiction) that the optimal $\{\Delta_i\}_{i=1}^n$ must satisfy $\Delta_i = 0$ for all $i\ge 2$. 
To show this, let us assume that the optimal solution $\{\rho_n,\tau,\{\Delta_i\}_{i=1}^n\}$ has $\Delta_j > 0$ for some $j\ge 2$. 
Let us now construct a new solution ($\{\rho_n', \tau',\{\Delta_i'\}_{i=1}^n\}$) and show that the new solution performs better than $\{\rho_n,\tau,\{\Delta_i\}_{i=1}^n\}$.
We construct the new solution as follows:
\begin{align} \label{eq:new_solution}
    \Delta_1' = \Delta_1 + \frac{1-\nu^{n-j+2}}{1-\nu^{n+1}}\Delta_j,\quad \Delta_j' = 0, \text{  and   } \Delta_i' = \Delta_i, \quad \forall i \notin \{1,j\}. 
\end{align}
Therefore, one may verify that the new solution results in
\begin{align}
    \tau' &=  t_f - \Big( \frac{1-\nu^{n}}{1-\nu}\rho_0 + \frac{1-\nu^{n}}{1-\nu}\Delta_1'  + \frac{1-\nu^{n-1}}{1-\nu}\Delta_2'  + \cdots +\Delta_n' \Big) \nonumber \\
    & = \tau - \frac{1-\nu^{n}}{1-\nu}\Delta_1' + \frac{1-\nu^{n}}{1-\nu}\Delta_1 + \frac{1-\nu^{n+1-j}}{1-\nu}\Delta_j \nonumber \\
    & \overset{\eqref{eq:new_solution}}{=} \tau - \frac{1-\nu^{n}}{1-\nu}\frac{1-\nu^{n-j+2}}{1-\nu^{n+1}}\Delta_j + \frac{1-\nu^{n+1-j}}{1-\nu}\Delta_j \nonumber \\
    &< \tau. \label{eq:tau}
\end{align}
One may further verify that
\begin{align} \label{eq:rho_n}
    \rho_n - \tau' &=  t_f - \Big( \frac{1-\nu^{n+1}}{1-\nu}\rho_0 + \frac{1-\nu^{n+1}}{1-\nu}\Delta_1'  + \frac{1-\nu^{n}}{1-\nu}\Delta_2'  + \cdots +\frac{1-\nu^2}{1-\nu}\Delta_n' \Big) \nonumber \\
    & \overset{\eqref{eq:new_solution}}{=} t_f - \Big(  \frac{1-\nu^{n+1}}{1-\nu}\rho_0 + \frac{1-\nu^{n+1}}{1-\nu}\Delta_1  + \frac{1-\nu^{n}}{1-\nu}\Delta_2  + \cdots +\frac{1-\nu^2}{1-\nu}\Delta_n \Big)\nonumber\\
    &=\rho_n - \tau.
\end{align}
Therefore, the new solution satisfies the constraints of the optimization problem, and furthermore,
\begin{align*}
    \nu \tau' + [\rho_n' -\tau']^+ < \nu \tau + [\rho_n - \tau]^+,
\end{align*}
where the inequality follows from \eqref{eq:tau} and \eqref{eq:rho_n}.
This concludes that the optimal solution must have $\Delta_i = 0$ for all $i\ge 2$.
To find optimal $\Delta_1$, we first write $\tau = t_f - \frac{1-\nu^n}{1-\nu}(\rho_0+\Delta_1)$ and $\rho_n = \nu^n (\rho_0+\Delta_1)$, and then rewrite the optimization problem \eqref{eq:optimization} as follows:
\begin{align} \label{eq:optimization2}
\begin{split}
    \min_{\Delta_1} \qquad \qquad & \nu\Big(t_f - \frac{1-\nu^n}{1-\nu}(\rho_0+\Delta_1)\Big) + \Big[\frac{1-\nu^{n+1}}{1-\nu}(\rho_0+\Delta_1) - t_f\Big]^+ \\
    \text{subject to  } & \Delta_1 \ge 0.
\end{split}
\end{align}
At this point one may verify that the objective function in \eqref{eq:optimization2} is piecewise linear with a negative slope in the interval $\Delta_1\in [0, \frac{1-\nu}{1-\nu^{n+1}}t_f - \rho_0]$ and with a positive slope in the interval $(\frac{1-\nu}{1-\nu^{n+1}}t_f - \rho_0, \infty)$.
Therefore, the optimal solution is 
\[
\Delta_1 = \frac{1-\nu}{1-\nu^{n+1}}t_f - \rho_0. 
\]
This concludes the proof.

\section{Proof of Lemma~3}\label{AP:Lemma3}
Given any $\rho_0$ and $t_f$, the pay-off for a pursuer with continuous sensing is $\phi([\rho_0 - (1-\nu)t_f]^+)$. 
To see this, recall that the pure pursuit is optimal when the pursuer has continuous sensing. 
Therefore, the final distance between the players at the end of the game is $[\rho_0 - (1-\nu)t_f]$ if $t_f < \frac{\rho_0 -\rcap}{1-\nu} $ and $0$ otherwise.
Thus, the final pay-off is $\phi([\rho_0 - (1-\nu)t_f]^+)$. 

On the other hand, from Theorem~1, we obtain\footnote{Since we assume $\nu \rho_0 > \sqrt{1 + \nu^2} \rcap$, the inequality in Theorem~1 becomes an equality} 
 \begin{align} \label{eq:Vcases}
     V(\rho_0, t_f, \ell) =  \begin{cases}
         0, \qquad \quad\quad\text{if } t_f \ge\!\!&\frac{1-\nu^{\ell +1}} {1-\nu}\rho_0, ~~\nu^{\ell +1}\rho_0 \le \rcap,\\
         \phi(\nu t_f+\rho_0 - t_f), &\text{if } t_f \le \frac{1-\nu^{\ell +1}} {1-\nu}\rho_0,\\
         \phi(\frac{1-\nu}{1-\nu^{\ell +1}}\nu^{\ell +1}t_f), & \text{otherwise}. 
     \end{cases}
 \end{align} 
At this point we shall find $\ell = n^*$ such that $V(\rho_0, t_f, n^*) = \phi([\rho_0 - (1-\nu)t_f]^+)$.
To find such $n^*$, we consider two separate cases: (i) $t_f < \frac{\rho_0 - \rcap}{1-\nu}$, and (ii) $t_f \ge \frac{\rho_0 - \rcap}{1-\nu}$.

\textbf{Case 1: $t_f < \frac{\rho_0 - \rcap}{1-\nu}$:} In this case $\rho_0 - (1-\nu)t_f > \rcap$ and hence, $\phi([\rho_0 - (1-\nu)t_f]^+) > 0$. 
From the three cases in \eqref{eq:Vcases}, one may readily verify that there does not exist any $n^*$ that satisfies the condition for the first case (i.e. $t_f \ge\!\!\frac{1-\nu^{n^* +1}} {1-\nu}\rho_0, ~~\nu^{n^* +1}\rho_0 \le \rcap$) under $t_f < \frac{\rho_0 - \rcap}{1-\nu}$.

Therefore, either condition for the second case must be satisfied, i.e., $t_f \le \frac{1-\nu^{n^* +1}} {1-\nu}\rho_0$ which yields  $n^* =  \big\lceil\frac{\log(\rho_0 - (1-\nu)t_f) - \log(\rho_0)}{\log (\nu)} - 1\big\rceil = \big\lfloor\frac{\log(\rho_0 - (1-\nu)t_f) - \log(\rho_0)}{\log (\nu)}\big\rfloor$, or the condition for third case must be satisfied along with $ \phi(\frac{1-\nu}{1-\nu^{n^* +1}}\nu^{n^* +1}t_f) = \phi([\rho_0 - (1-\nu)t_f])$.
Now we shall show that the third case does not always result into a valid $n^*$. 
To see this, note that, due to Assumption~1, $ \phi(\frac{1-\nu}{1-\nu^{n^* +1}}\nu^{n^* +1}t_f) = \phi([\rho_0 - (1-\nu)t_f])$ implies $ \frac{1-\nu}{1-\nu^{n^* +1}}\nu^{n^* +1}t_f= \rho_0 - (1-\nu)t_f$, which further leads to 
\[
\nu^{n^*+1} = \frac{\rho_0 - (1-\nu)t_f}{\rho_0}.
\]
Since $n^*$ must be an integer, the above equality does not hold for every choice of $\rho_0, t_f$ and $\nu$. 
In the rare case when this equality holds, we obtain $n^* = \frac{\log(\rho_0 - (1-\nu)t_f) - \log(\rho_0)}{\log (\nu)} - 1 \le  \big\lfloor\frac{\log(\rho_0 - (1-\nu)t_f) - \log(\rho_0)}{\log (\nu)}\big\rfloor$.\\

\textbf{Case 2: $t_f \ge \frac{\rho_0 - \rcap}{1-\nu}$:} In this case, the pay-off for the continuous sensing case is $0$. Comparing this with \eqref{eq:Vcases}, this can be achieved in three ways:

(i) Find $n^*$ such that $t_f \ge\!\!\frac{1-\nu^{n^* +1}} {1-\nu}\rho_0, ~~\nu^{n^* +1}\rho_0 \le \rcap$ along with $t_f \ge \frac{\rho_0 - \rcap}{1-\nu}$ \\

(ii)  Find $n^*$ such that $t_f \le \!\!\frac{1-\nu^{n^* +1}} {1-\nu}\rho_0$ along with $t_f \ge \frac{\rho_0 - \rcap}{1-\nu}$ \\

(iii)  Find $n^*$ such that $t_f \ge\!\!\frac{1-\nu^{n^* +1}} {1-\nu}\rho_0, ~~\nu^{n^* +1}\rho_0 > \rcap$ along with $t_f \ge \frac{\rho_0 - \rcap}{1-\nu}$ and $\frac{1-\nu}{1-\nu^{n^* +1}}\nu^{n^* +1}t_f \le \rcap$.

One may verify that the third condition is infeasible since $\nu^{n^* +1}\rho_0 > \rcap$ and $\frac{1-\nu}{1-\nu^{n^* +1}}\nu^{n^* +1}t_f \le \rcap$ lead to $t_f < \!\!\frac{1-\nu^{n^* +1}} {1-\nu}\rho_0$ whereas the first inequality under this condition condition requires $t_f \ge\!\!\frac{1-\nu^{n^* +1}} {1-\nu}\rho_0$.

Both (i) and (ii) simplifies to $n^* \ge \frac{\log(\rcap) - \log(\rho_0)}{\log (\nu)} -1$ or $n^* = \big\lfloor\frac{\log(\rcap) - \log(\rho_0)}{\log (\nu)}\big\rfloor$.
This concludes the proof.

\end{document}